\documentclass[10pt]{article}
\usepackage[authoryear,round]{natbib}						% need for \cite
\usepackage{natbib}																	% need for \cite
\usepackage{amsmath}    													% need for subequations
\usepackage{graphicx}   													% need for figures
\usepackage{subfigure}														% need for subfigures
\usepackage{amssymb}    													% gives you \mathbb{} font
\usepackage[mathscr]{eucal}												% gives you \mathscr font
\usepackage{cancel}																% gives you the ability to visibly cross out terms in equations}
\usepackage[normalem]{ulem} 																% gives you \sout
\usepackage{pstricks}
\usepackage{rotating}
\usepackage{lscape}
\usepackage[paperwidth=8.5in,paperheight=11in,top=1.25in, bottom=1.25in, left=1.00in, right=1.00in]{geometry}
\usepackage{mathtools}														% need for `show only references'
\mathtoolsset{showonlyrefs=true}									% only equations which are labeled AND referenced will be numbered.
\usepackage{amsthm}																% need for theorem-proof environment
\theoremstyle{plain}
\newtheorem{theorem}{Theorem}
\newtheorem{proposition}[theorem]{Proposition}	

\theoremstyle{definition}
\newtheorem{definition}[theorem]{Definition}

\newtheorem{remark}[theorem]{Remark}

\renewcommand{\(}{\left(}				
\renewcommand{\)}{\right)}
\renewcommand{\[}{\left[}
\renewcommand{\]}{\right]}
\newcommand{\norm}[1]{\left\|{#1}\right\|}		
\newcommand{\dom}[1]{{\rm dom}(#1)}		

\newcommand\Cb{\mathbb{C}}	
\newcommand\Eb{\mathbb{E}}			
\newcommand\Pb{\mathbb{P}}	
							
\newcommand\Rb{\mathbb{R}}										
\newcommand\Ib{\mathbb{I}}

\newcommand\Ac{\mathscr{A}}
\newcommand\Bc{\mathscr{B}}

\newcommand\Fc{\mathscr{F}}

\newcommand\Hc{\mathscr{H}}

\newcommand\Oc{\mathscr{O}}
\newcommand\Pc{\mathscr{P}}

\newcommand\Sc{\mathscr{S}}

\newcommand\eps{\varepsilon}

\newcommand\Om{\Omega}
\newcommand\sig{\sigma}

\newcommand\lam{\lambda}
\newcommand\del{\delta}

\renewcommand\d{\partial}

\begin{document}
%\doi{10.1080/1469768YYxxxxxxxx}
 %\issn{1469-7696} \issnp{1469-7688} \jvol{00} \jnum{00} \jyear{2008} \jmonth{July}

%\markboth{Taylor \& Francis and A. Consultant}{\LaTeXe\ guide for authors}

\title{The Exact Smile of certain Local Volatility Models}

\author{
Matthew Lorig
\thanks{ORFE Department, Princeton University, Princeton, USA.  Work partially supported by NSF grant DMS-0739195.}
}

\date{This version: \today}

\maketitle

\begin{abstract}
We introduce a new class of local volatility models.  Within this framework, we obtain expressions for both (i) the price of any European option and (ii) the induced implied volatility smile.  As an illustration of our framework, we perform specific pricing and implied volatility computations for a CEV-like example.  Numerical examples are provided.
\end{abstract}

\textbf{keywords}:
CEV, local volatility, stochastic volatility, implied volatility.
%\end{keywords}

%\begin{AMS}
%35P10, 47A55, 60H30, 62P05.
%\end{AMS}

%%%%%%%%%%%%%%%%%%%%%%%%%%%%%%%%%%%%%%%
%
%		SECTION: Introduction
%
%%%%%%%%%%%%%%%%%%%%%%%%%%%%%%%%%%%%%%%

\section{Introduction}
\label{sec:intro}
\emph{Local volatility} models are a class of equity models in which the volatility $\sig_t$ of an asset $X$ is described by a function of {}{time} and the present level of $X$.  That is, $\sig_t = \sig({}{t},X_t)$.  Among local volatility models, perhaps the most well-known is the constant elasticity of variance (CEV) model of \cite{CoxCEV}.  An extension of the CEV model to defaultable assets (the Jump-to-Default CEV or JDCEV model) is derived in \cite{JDCEV}.  One advantage of these two local volatility models is that they admit closed-form pricing formulas for European options, written as infinite series of special functions.  {}{Another advantage of local volatility models is that, for models whose transition density is not available in closed form, accurate density approximations are often available.  See, for example, \cite{pagliarani2011analytical}.}
\par
In this paper, we introduce a new class of local volatility models which, like the CEV and JDCEV models, allow for European option prices to be expressed in closed form as an infinite series.  Additionally, we derive an expression for the \emph{exact} implied volatility surface induced by our class of models.  Previous studies of the implied volatility surface induced by local volatility models focused on heat-kernel expansions to derive \emph{asymptotic approximations} of the volatility smile (see e.g., \cite{gatherallocal,henry2005general} and references therein).  It is worth mentioning that \cite{dupire1994pricing} solves the inverse problem of finding a formula for the local volatility function the produces a given observed implied volatility surface exactly.
\par
{}{Essential for our mathematical presentation, is the use of spectral theory.  The spectral representation theorem has been widely applied in mathematical finance.   An exhaustive review would be prohibitive.  However, we mention the seminal work of \cite{linetsky2003}, who lay the groundwork for option-pricing with eigenfunctions in a scalar diffusion setting.  For applications of eigenfunction methods in a stochastic volatility (i.e., multivariate) setting, we refer the reader to \cite{lorig2,lorigmultiscale}.  While previous spectral-related work has focused exclusively on eigenfunction expansions for \emph{self-adjoint} operators in Hilbert space, here, we focus on generalized eigenfunction expansions for \emph{normal} operators.  To our knowledge, this is the first time the spectral theory of normal operators has been used in a financial setting.}
\par
The rest of this paper proceeds as follows: in section \ref{sec:assumptions} we present our class of models and describe our assumptions about the market.  In section \ref{sec:pricing} we derive a formula for the price of a European option, written in a general form which is valid for any model within our framework.  In section \ref{sec:impvol} we provide an formula for the implied volatility smile induced by our class models.  In section \ref{sec:example}, as an example of our framework, we perform explicit pricing and implied volatility computations for a CEV-like example.  Numerical results are provided at the conclusion of the text.  An appendix with some mathematical background is also provided.  Concluding remarks can be found in section \ref{sec:conclusion}.

%%%%%%%%%%%%%%%%%%%%%%%%%%%%%%%%%%%%%%%
%
%		SECTION: Model and assumptions
%
%%%%%%%%%%%%%%%%%%%%%%%%%%%%%%%%%%%%%%%

\section{Model and assumptions}
\label{sec:assumptions}
We assume a frictionless market, no arbitrage and take an equivalent martingale measure $\Pb$ chosen by the market on a complete filtered probability space $(\Om,\Fc,\{\Fc_t,t \geq 0\},\Pb)$.  The filtration $\{\Fc_t, t \geq 0 \}$ represents the history of the market.  All processes defined below live on this space.  For simplicity we assume zero interest rates and no dividends so that all assets are martingales.  We consider an asset $X$ whose dynamics are given by
\begin{align}
dX_t
	&=		\( a^2 + \eps \, \eta( \log X_t) \)^{1/2} X_t \, dW_t , \label{eq:dX}
\end{align}
where, $a>0$, $\eps \geq 0$, the function $\eta:\Rb \to \Rb^+$ is an element of $\Sc$ (the Schwartz space of rapidly decreasing functions on $\Rb$; see equation \eqref{eq:schwartz} for a definition) and $W$ is a Brownian motion.  
{}{The restriction $\eta \in \Sc$ is needed to prove Theorem \ref{thm:u.eps}}.
We assume that $X_0>0$, the initial value of $X$ is known.  Note that $X$ has \emph{local} volatility $\sig(X_t) = \( a^2 + \eps \, \eta( \log X_t) \)^{1/2}$.  Obviously, if $\eta=0$ then $X$ is a geometric Brownian motion.  This will be key for the implied volatility analysis in section \ref{sec:impvol}.  Observe that both zero and infinity are natural boundaries according to Feller's boundary classification for one-dimensional diffusions (see \cite{borodin} pp. 14-15).  That is, both zero and infinity are unattainable.
\par
In what follows it will be convenient to introduce $Y = \log X$.  A simple application of It\^o's formula shows that $Y$ satisfies
\begin{align}
dY_t
	&=		-\frac{1}{2} \( a^2 + \eps \, \eta(Y_t) \) dt + \( a^2 + \eps \, \eta(Y_t) \)^{1/2} dW_t . \label{eq:dY}
\end{align}
{}{With $\eta \in \Sc$, the volatility and drift coefficients in \eqref{eq:dY} satisfy the usual growth and Lipschitz conditions, which guarantee a unique strong solution to SDE \eqref{eq:dY}.  See \cite{oksendal} Theorem 5.2.1.}

%%%%%%%%%%%%%%%%%%%%%%%%%%%%%%%%%%%%%%%
%
%		SECTION: Option Pricing
%
%%%%%%%%%%%%%%%%%%%%%%%%%%%%%%%%%%%%%%%

\section{Option pricing}
\label{sec:pricing}
We wish to find the time-zero value of a European-style option with payoff $h(Y_t)$ at time $t>0$.  Using risk-neutral pricing we express the initial value of the option $u^\eps(t,y)$ as the risk-neutral expectation of the option payoff 
\begin{align}
u^\eps(t,y) 
	&=		\Eb_y \, h(Y_t) , &
	&{}{(r=0)} \label{eq:u.eps.def}
\end{align}
where the notation $\Eb_y$ indicates $\Pb$-expectation starting from $y = \log X_0$.  
{}{Suppose $h \in C_0^2(\Rb)$ (compactly supported functions with continuous derivatives up to order 2).}
Then, the function $u^\eps(t,y)$ satisfies the Kolmogorov backward equation
\begin{align}
\( - \d_t + \Ac^\eps \) u^\eps
	&=		0 , &
u^\eps(0,y)
	&=		h(y) . \label{eq:u.eps.PIDE}
\end{align}
where $\Ac^\eps$ is the generator of the process $Y$.  The domain of $\Ac^\eps$ is defined as the set of $f$ for which the limit $\lim_{t \to 0}\frac{1}{t}\( \Eb_y f(Y_t) - f(y) \)$ exists in the strong sense.  For any {}{$f \in C_0^2(\Rb)$} the generator $\Ac^\eps$ has the explicit representation
\begin{align}
\Ac^\eps
	&=		\Ac_0 + \eps \, \eta \, \Ac_1 , &
\Ac_0
	&= 		\frac{1}{2} a^2 \( \d^2 - \d \)  , &
\Ac_1
	&=		\frac{1}{2} \( \d^2 - \d \) , 
	%&
%\dom{\Ac_i}
	%&=		\Sc ,
\label{eq:A0.A1}
\end{align}
where $\d$, without the subscript $t$, indicates differentiation with respect to $y$.
\begin{remark}\rm
\label{rmk:call}
{}{It is possible to extend our results to payoff functions $h$ that are continuous and have linear growth in $\log y$ (e.g. Call options).  However, rigorous justification for this it outside the scope of this paper.  Numerical tests are provided to support this claim.}
\end{remark}
\begin{remark}\rm
The operators $\Ac_0$ and $\Ac_1$ are normal operators in the Hilbert space $\Hc=L^2(\Rb,dy)$ and satisfy the following (improper) eigenvalue equations (neither $\Ac_0$ nor $\Ac_1$ have any proper eigenvalues)
\begin{align}
\Ac_0 \psi_\lam
	&=		\phi_\lam \psi_\lam , &
\psi_\lam
	&=		\frac{1}{\sqrt{2\pi}} e^{i \lam y} , &
\phi_\lam
	&=		\frac{1}{2} a^2 \( -\lam^2 - i \lam \) , \\
\Ac_1 \psi_\lam
	&=		\chi_\lam \psi_\lam , &
\psi_\lam
	&=		\frac{1}{\sqrt{2\pi}} e^{i \lam y} , &
\chi_\lam
	&=		\frac{1}{2} \( -\lam^2 - i \lam \) . \label{eq:eigen.A1}
\end{align}
{}{Note that, as shown by \cite{diracdelta}, for any analytic $f$, we have
\begin{align}
\frac{1}{2\pi} \int_\Rb dx \, e^{-i\lam x} 
	&=	\del(\lam) , &
\int_\Rb d\mu \, \del(\lam-\mu) f(\mu) 
	&=	f(\lam) , &
\lam
	&\in \Cb .
\end{align}
Thus, the generalized eigenfunctions satisfy the orthogonality relation
\begin{align}
( \psi_\lam , \psi_\mu )
	&=	\frac{1}{2\pi} \int_\Rb dx \, e^{-i ( \lam - \mu) x}
	=		\del(\lam - \mu ) , &
\lam,\mu
	&\in \Cb . \label{eq:diracdelta}
\end{align}
See also, \cite{friedman1956principles}, equation (4.35).}
Note also that Borel-measurable functions of normal operators (e.g., $g(\Ac_0)$) are well-defined by the spectral theorem for normal operators, as explained in Appendix \ref{sec:spectral}.
\end{remark}
\par
We seek a solution to Cauchy problem \eqref{eq:u.eps.PIDE} of the form
\begin{align}
u^\eps
	&=		\sum_{n=0}^\infty \eps^n \, u_n . \label{eq:u.eps.expand}
\end{align}
We will justify this expansion in Theorem \ref{thm:u.eps}.  Inserting the expansion \eqref{eq:u.eps.expand} into Cauchy problem \eqref{eq:u.eps.PIDE} and collecting terms of like powers of $\eps$ we obtain
\begin{align}
\Oc(1):&&
( - \d_t + \Ac_0 ) u_0
	&=		0  , &
u_0(0,y)
	&=		h(y) , \\
\Oc(\eps^n):&&
( - \d_t + \Ac_0 ) u_n
	&=		- \eta \Ac_1 u_{n-1} , &
u_n(0,y)
	&=		0  . \label{eq:un.pde}
\end{align}
The solution to the above equations is
\begin{align}
\Oc(1):&&
u_0(t,y)
	&=	e^{t\Ac_0} h(y) , \\
\Oc(\eps^n):&&
u_n(t,y)
	&=		\int_0^t ds \, e^{(t-s)\Ac_0} \eta(y) \Ac_1 u_{n-1} (s,y) .
\end{align}
Using the equation \eqref{eq:operator} from appendix \ref{sec:spectral} to write the spectral representation of $e^{t \Ac_0}$ we obtain
\begin{align}
\Oc(1):&&
u_0(t,y)
	&=		\int_\Rb d\lam \, e^{t \phi_\lam} (\psi_\lam, h) \psi_\lam(y) , \\
\Oc(\eps^n):&&
u_n(t,y)
	&=		\int_0^t \int_\Rb ds \, d\mu \, e^{(t-s)\phi_\mu} \( \psi_\mu , \eta \Ac_1 u_{n-1} (s,\cdot) \) \psi_\mu(y) ,
\end{align}
After a bit of algebra, we find an explicit representation for $u_n(t,y)$ 
\begin{align}
u_n
	&=		\underbrace{ \int \cdots \int}_{n+1} \( \prod_{k=0}^n d\lam_k \)
				\( \sum_{k=0}^n \frac{e^{t \phi_{\lam_k}}}{\prod_{j \neq k}^n (\phi_{\lam_k}-\phi_{\lam_j})}\)
				\( \prod_{k=0}^{n-1} \( \psi_{\lam_{k+1}}, \eta \Ac_1 \psi_{\lam_k} \) \) \(\psi_{\lam_0},h \) \, \psi_{\lam_n} . 
				\label{eq:u.n}
\end{align}
\begin{remark}
{}{As we will show in section \ref{sec:example}, for certain choices of $\eta$, the $(n+1)$-fold integral in \eqref{eq:u.n} will collapse into a single integral.}
\end{remark}
We have now obtained a formal expansion \eqref{eq:u.eps.expand}-\eqref{eq:u.n} for the price of a European option.  The following theorem provides conditions under which the expansion is guaranteed to be valid.  
\begin{theorem}
\label{thm:u.eps}
Suppose $\eps \leq \frac{a^2}{\left\| \eta \right\|}$, where $\left\| \eta \right\|=\sqrt{(\eta,\eta)}$.  Then the option price $u^\eps(t,y)$ is given by \eqref{eq:u.eps.expand}-\eqref{eq:u.n}.
\end{theorem}
\begin{proof}
See Appendix \ref{sec:proof3}.
\qquad \end{proof}

%%%%%%%%%%%%%%%%%%%%%%%%%%%%%%%%%%%%%%%
%
%		SECTION: Implied Volatility
%
%%%%%%%%%%%%%%%%%%%%%%%%%%%%%%%%%%%%%%%

\section{Implied volatility}
\label{sec:impvol}
In this section we fix $(t,y)$ and a call option payoff $h(y)=(e^y-e^k)^+$.  Note that
\begin{align}
(\psi_\lam, h)
	&=		\frac{-e^{k-i k \lam}}{\sqrt{2 \pi } \(i \lam + \lam^2 \)} , &
\text{Im}(\lam)
	&< 		- 1 .
\end{align}
The following definitions will be useful:
\begin{definition}\rm
\label{def:bs}
The \emph{Black-Scholes Price} $u^{BS}:\Rb^+ \to \Rb^+$ is defined as
\begin{align}
u^{BS}(\sig)
	&:=		\int_\Rb d\lam \, e^{t \phi^{BS}_\lam(\sig)} (\psi_\lam, h) \psi_\lam , &
\phi^{BS}_\lam(\sig)
	&=		\frac{1}{2}\sig^2(-\lam^2 - i \lam) . \label{eq:u.BS}
\end{align}
\end{definition}
\begin{remark}
\label{rmk:bs}
\rm
{}{Usually, the Black-Scholes price is written
\begin{align}
u^{BS}(\sig)
	&=	\int_\Rb dx \, h(x) \, \Phi_{m,s^2}(x) , &
m
	&=	y-\frac{1}{2}\sig^2 t , &
s^2
	&=	\sig^2 t , \label{eq:BS.price}
\end{align}
where $\Phi_{m,s^2}$ is a Gaussian density with mean $m$ and variance $s^2$.  Equation \eqref{eq:u.BS} is simply the Fourier representation of \eqref{eq:BS.price}.}
\end{remark}
\begin{definition}
\label{def:imp.vol}
\rm
{}{For each fixed $\log$ spot $y$, time to maturity $t$, and $\log$ strike price $k$}, the \emph{Implied Volatility} is defined implicitly as the unique number $\sig^\eps \in \Rb^+$ such that
\begin{align}
u^{BS}(\sig^\eps)
	&=		u^\eps , \label{eq:imp.vol.def}
\end{align}
where $u^\eps$ is as given in Theorem \ref{thm:u.eps}.
\end{definition}
\begin{remark}\rm
Note that $u_0 =u^{BS}(a)$.  As shown in \cite{lorig2012impvol}, when $u^\eps$ can be expanded as an analytic series whose first term corresponds to $u^{BS}$, one can obtain the exact implied volatility corresponding to $u^\eps$.
\end{remark}
\begin{remark}
\label{rmk:exist}
\rm
For $0 < t < \infty$ the existence and uniqueness of the implied volatility $\sig^\eps$ can be deduced by using the general arbitrage bounds for call prices and the monotonicity of $u^{BS}$.
{}{See \cite{fpss}, Section 2.1, Remark (i).}
\end{remark}
\begin{remark}
\label{rmk:analytic}
{}{Observe that, for any $\sig_0>0$ and $\sig_0 + \del>0$, the function $u^{BS}(\sig_0 + \del)$ is given by its Taylor series:
\begin{align}
u^{BS}(\sig_0 + \del)
	&=	\sum_{n=0}^\infty \frac{\del^n}{n!} \d_\sig^n u^{BS}(\sig_0), &
\d_\sig^n u^{BS}(\sig_0)
	&=		\int_\Rb d\lam \, \( \d_\sig^n e^{t \phi^{BS}_\lam(\sig_0)} \) ( \psi_\lam, h) \psi_\lam .
\end{align}
Observe also that, by monotonicity of $u^{BS}$ we have $\d_\sig u^{BS}(\sig)> 0$ for all $\sig>0$.}  Therefore, $u^{BS}$ is an invertible analytic function.  By the Lagrange inversion theorem, the inverse function $[u^{BS}]^{-1}$ is also analytic.
\end{remark}
\par
Clearly, $u^\eps$ is an analytic function of $\eps$ (we derived its power series expansion).  It is a useful fact that the composition of two analytic functions is also analytic (see \cite{brown1996complex}, section 24, p. 74).  Thus, in light of Remark \ref{rmk:analytic}, we deduce that $\sig^\eps = [u^{BS}]^{-1}(u^\eps)$ is an analytic function and therefore has a power series expansion in $\eps$.  We write this expansion as follows
\begin{align}
\sig^\eps
	&=		\sig_0 + \del^\eps , &
\del^\eps
	&=		\sum_{k=1}^\infty \eps^k \sig_k  . \label{eq:sigma.expand}
\end{align}
Taylor expanding $u^{BS}$ about the point $\sig_0$ we have
\begin{align}
u^{BS}(\sig^\eps)
	&=	u^{BS}(\sig_0 + \del^\eps) \\
	&=		\sum_{n=0}^\infty \frac{1}{n!}(\del^\eps \d_\sig )^n u^{BS}(\sig_0) \\
	&=		u^{BS}(\sig_0) +
						\sum_{n=1}^\infty \frac{1}{n!} \( \sum_{k=1}^\infty \eps^k \sig_k \)^n \d_\sig^n u^{BS}(\sig_0) \\
	&=		u^{BS}(\sig_0) + 
						\sum_{n=1}^\infty \frac{1}{n!}  
						\[ \sum_{k=1}^\infty \( \sum_{j_1+\cdots+j_n=k} \prod_{i=1}^n \sig_{j_i} \) \eps^k \] \d_\sig^n u^{BS}(\sig_0) \\
	&=		u^{BS}(\sig_0) +
						\sum_{k=1}^\infty \eps^k 
						\[ \sum_{n=1}^\infty \frac{1}{n!} \( \sum_{j_1+\cdots+j_n=k} \prod_{i=1}^n \sig_{j_i} \) \d_\sig^n \] u^{BS}(\sig_0) \\
	&=	u^{BS}(\sig_0) +
						\sum_{k=1}^\infty \eps^k 
						\[ \sig_k \d_\sig + \sum_{n=2}^\infty \frac{1}{n!}\( \sum_{j_1+\cdots+j_n=k} \prod_{i=1}^n \sig_{j_i} \)  \d_\sig^n \]
						u^{BS}(\sig_0) .			\label{eq:u.bs.expand}
\end{align}
Now, we insert expansions \eqref{eq:u.eps.expand} and \eqref{eq:u.bs.expand} into \eqref{eq:imp.vol.def} and collect terms of like order in $\eps$
\begin{align}
\Oc(1):&&
u_0
	&= 		u^{BS}(\sig_0) , \\
\Oc(\eps^k):&&
u_k
	&=		\sig_k \d_\sig u^{BS}(\sig_0)
				+ \sum_{n=2}^\infty \frac{1}{n!}\( \sum_{j_1+\cdots+j_n=k} \prod_{i=1}^n \sig_{j_i} \)  \d_\sig^n	u^{BS}(\sig_0) , &
k
	&\geq	1 .
\end{align}
Solving the above equations for $\{\sig_k\}_{k=0}^\infty$ we find
\begin{align}
\Oc(1):&&
\sig_0
	&=	a , \\
\Oc(\eps^k):&&
\sig_k
	&=	\frac{1}{\d_\sig u^{BS}(\sig_0)}
			\( u_k - 
			\sum_{n=2}^\infty \frac{1}{n!}\( \sum_{j_1+\cdots+j_n=k} \prod_{i=1}^n \sig_{j_i} \)  \d_\sig^n	u^{BS}(\sig_0)
			\) , &
k
	&\geq 1 . \label{eq:sig.k}
\end{align}
\begin{remark}\rm
\label{rmk:easy}
The right hand side of \eqref{eq:sig.k} involves only $\sig_j$ for $j \leq k-1$.  Thus, the $\{\sig_k\}_{k=1}^\infty$ can be found recursively.
\end{remark}
Explicitly, up to $\Oc(\eps^4)$ we have
\begin{align}
\Oc(\eps):&&
\sig_1
	&= 	\frac{u_1}{\d_\sig u_0}, \\
\Oc(\eps^2):&&
\sig_2
	&= 	\frac{u_2 - \tfrac{1}{2} \sig_1^2 \d_\sig^2 u_0}{\d_\sig u_0}, \\
\Oc(\eps^3):&&
\sig_3
	&= 	\frac{u_3 - (\sig_2 \sig_1 \d_\sig^2 + \tfrac{1}{6}\sig_1^3 \d_\sig^3) u_0}{\d_\sig u_0}, \\
\Oc(\eps^4):&&
\sig_4
	&=  \frac{u_4 - (\sig_3 \sig_1 \d_\sig^2 + \tfrac{1}{2} \sig_2^2 \d_\sig^2 
									+ \tfrac{1}{2} \sig_2 \sig_1^2 \d_\sig^3 + \tfrac{1}{24} \sig_1^4 \d_\sig^4) u_0}{\d_\sig u_0}.
\end{align}
We summarize our implied volatility result in the following theorem:
\begin{theorem}
\label{thm:imp.vol}
The implied volatility $\sig^\eps$ defined in \eqref{eq:imp.vol.def} is given explicitly by \eqref{eq:sigma.expand} where $\sig_0=a$ and $\{\sig_k\}_{k=1}^\infty$ are given by \eqref{eq:sig.k}.
\end{theorem}
\begin{remark}\rm
Everything we have done so far is \underline{exact}.  
The accuracy of the implied volatility expansion \eqref{eq:sigma.expand} is limited only by the number of terms one wishes to compute.
\end{remark}

%%%%%%%%%%%%%%%%%%%%%%%%%%%%%%%%%%%%%%%%%%%%%%%%%%%%
%
%			Example
%
%%%%%%%%%%%%%%%%%%%%%%%%%%%%%%%%%%%%%%%%%%%%%%%%%%%%

\section{CEV-like example}
\label{sec:example}
In the constant elasticity of variance (CEV) model of \cite{CoxCEV} the dynamics of $X$ are assumed to be of the form $dX_t = \sqrt{\eps}\, X_t^{\beta/2} X_t dW_t$.  A key feature of the CEV model is that, when $\beta < 0$, the local volatility function $\sig(x) = \sqrt{\eps}\, x^{\beta/2}$ \emph{increases} as $x \searrow 0$, which (i) is consistent with the leverage effect and (ii) results in a negative implied volatility skew.  However, values of $\beta<0$ also cause the volatility to drop unrealistically close to zero as $x$ increases.  If we choose $\eta(y)=e_\beta(y):=e^{\beta y}$ then from \eqref{eq:dX} the dynamics of $X$ become
\begin{align}
dX_t
	&=		( a^2 + \eps X_t^\beta )^{1/2} X_t dW_t , \label{eq:dX.2}
\end{align}
Note that the local volatility function $\sig(x)=(a^2 + \eps \, x^{\beta})^{1/2}$ behaves asymptotically like $\sig(x) \sim \sqrt{\eps}\, x^{\beta/2}$ as $x \searrow 0$ and behaves asymptotically like a constant $\sig(x) \sim a$ as $x \nearrow \infty$.
\begin{remark}\rm
\label{rmk:domain}
Because $e^{\beta y}$ is unbounded as $y \to -\infty$ (recall $\beta<0$), the function $e_\beta \notin \Sc$.  However, we can modify the domain of $u^\eps(t,x)$ to be $\Rb^+ \times \Rb_0$ where $\Rb_0:=(y_0,\infty)$ {}{and $y_0 \in \Rb$ is arbitrary}.  The operators $\Ac_0$ and $\Ac_1$ would then be defined on $L^2(\Rb_0,dy)$ and the domain of these operators would include an absorbing boundary condition at $y_0$ (signifying default of $X$ the first time $X$ reaches the level $e^{y_0}$).  Note that $\left\| e_\beta \right\|_0 := ( \int_{y_0}^\infty | e_\beta |^2 dy )^{1/2}= e^{\beta y_0}/\sqrt{-2 \beta}$.  In the analysis that follows, it will simplify computations considerably if we continue to work on $L^2(\Rb,dy)$ as working on $L^2(\Rb_0,dy)$ would require modifying the eigenfunctions $\psi_\lam$ from complex exponentials $\exp(i\lam y)$ to sines $\sin (\lam y)$.  However, the simplification comes at a cost; in light of the conditions of theorem \eqref{thm:u.eps} our results may not be valid for values of {}{$y<\frac{1}{\beta} \log \frac{a^2 \sqrt{-2\beta}}{\eps}$}.
\end{remark}
\par
We wish to find a simplified expression for $u_n$ \eqref{eq:u.n} for the case $\eta=e_\beta$.  Using \eqref{eq:eigen.A1} and {}{\eqref{eq:diracdelta} we note that
\begin{align}
(\psi_\mu , e_\beta \Ac_1 \psi_\lam )
	&=	\chi_\lam \, \frac{1}{2 \pi} \int dx \, e^{i(\lam - \mu - i \beta)x} 
	=		\chi_\lam \, \del(\lam-\mu-i\beta) . \label{eq:delta}
\end{align}
}
Thus, the $(n+1)$-fold integral \eqref{eq:u.n} collapses into a single integral
\footnote{For a Dirac delta function with a complex argument $\zeta$ we have the following identity from \cite{diracdelta}: $\int_\Rb f(\lam) \del(\lam-\zeta) d\lam = f(\zeta)$.}
\begin{align}
u_n
	&=		\int_\Rb d\lam \( \sum_{k=0}^n \frac{e^{t \phi_{\lam-ik\beta}}}
				{\prod_{j\neq k}^n (\phi_{\lam-ik\beta}-\phi_{\lam-ij\beta})}\)
				\( \prod_{k=0}^{n-1} \chi_{\lam-ik\beta}\) (\psi_\lam, h) \, \psi_{\lam-in\beta} \\
	&=		e_{n\beta} \int_\Rb d\lam \( \sum_{k=0}^n \frac{e^{t \phi_{\lam-ik\beta}}}
				{\prod_{j\neq k}^n (\phi_{\lam-ik\beta}-\phi_{\lam-ij\beta})}\)
				\( \prod_{k=0}^{n-1} \chi_{\lam-ik\beta}\) (\psi_\lam, h) \, \psi_{\lam} .
				\label{eq:u.n.2}
\end{align}
\begin{remark}\rm
\label{rmk:truncation}
Although we have written the option price as an infinite series \eqref{eq:u.eps.expand}, from a practical standpoint, one is only able to compute $u^\eps \approx {}{ u^{(N)} :=} \sum_{n=0}^N \eps^n u_n$ for some finite $N$.  For any finite $N$ we may pass the sum $\sum_{n=0}^N$ through the integral appearing in \eqref{eq:u.n.2}.  Thus, for the purposes of computation, the most convenient way express the approximate option price is
\begin{align}
u^\eps
	&\approx	{}{u^{(N)}
	=}
				\int_\Rb d\lam \, (\psi_\lam, h) \, \psi_{\lam} \sum_{n=0}^N 
				\eps^n \, e_{n\beta} \( \sum_{k=0}^n \frac{e^{t \phi_{\lam-ik\beta}}}
				{\prod_{j\neq k}^n (\phi_{\lam-ik\beta}-\phi_{\lam-ij\beta})}\)
				\( \prod_{k=0}^{n-1} \chi_{\lam-ik\beta}\) . \label{eq:uN.approx}
\end{align}
Note, to obtain the approximate value of $u^\eps$, \emph{only a single integration is required}.  This makes our pricing formula as efficient as other models in which option prices are expressed as a Fourier-type integral (e.g. exponential L\'evy processes, Heston model, etc.).
\end{remark}
\subsection*{Numerical Results}
{}{In light of Remarks \ref{rmk:call}, \ref{rmk:domain} and \ref{rmk:truncation}, we provide some numerical tests supporting the use of the model considered in section \ref{sec:example}.}
\subsubsection*{Monte Carlo Test}
{}{To text the accuracy of approximation \eqref{eq:uN.approx}}, we compute the price of a series of European call options using approximation \eqref{eq:uN.approx} with $N=10$.  We then compute the price of the same series of call options by means of a Monte Carlo simulation using a standard Euler scheme with a time step of $10^{-3}$ years and $10^7$ sample paths.  The largest relative error obtained in the Monte Carlo simulations (i.e., standard error divided by price) was $0.0012$.  Finally, we convert call prices to implied volatilities by inverting Black-Scholes numerically.  The results of this procedure are plotted in figure \ref{fig:MonteCarlo}.  The implied volatilities resulting from the two methods of computation are indistinguishable.
\subsubsection*{Convergence of Transition Density}
Define the transition density $p^\eps(t,y;y_0)$ and the $\Oc(\eps^n)$ approximation of the transition density $p^{(n)}(t,y;y_0)$, which are obtained by setting the payoff function $h$ equal to a Dirac delta function $h = \del_y$.  Explicitly
\begin{align}
p^\eps(t,y;y_0)
	&=		\Eb_{y_0} \, \del_y(Y_t) , &
p^{(n)}(t,y;y_0)
	&=		\sum_{k=0}^n \eps^k p_k(t,y;y_0) .
\end{align}
{}{In order to test the rate of convergence of $p^{(n)}$ to $p^\eps$},
in figure \ref{fig:density}, we plot the approximate transition density $p^{(n)}$ for different values of $n$.  For $n \geq 6$ we see virtually no difference between $p^{(n)}$ and $p^{(n-1)}$.
\subsubsection*{Convergence of Implied Volatility}
Finally, {}{to see how well the implied volatility expansion of section \ref{sec:impvol} performs}, we define the $\Oc(\eps^n)$ approximation of the implied volatility
\begin{align}
\sig^{(n)}
	&:=		\sum_{k=0}^n \eps^k \sig_k , \label{eq:sigma.approximate}
\end{align}
where the $\sig_k$ are given by \eqref{eq:sig.k}.  In figure \ref{fig:impvol} we provide a numerical example illustrating convergence of $\sig^{(n)}$ to $\sig^\eps$.  {}{We compute $\sig^\eps$ using a two-step procedure.  First, we approximate $u^\eps$ using \eqref{eq:uN.approx} with $N=10$.  In light of the Monte Carlo simulation above, this should introduce almost no error.  Then, to find $\sig^\eps$, we solve $u^{BS}(\sig^\eps)=u^\eps$ numerically.}  Implied volatility is plotted as a function of the $\log$-moneyness to maturity ratio, $\text{LMMR}:=(k-y)/t$.  Convergence is fastest for values of $k$ near $y$ and slows as $k$ moves away from $y$ in the negative direction.

%%%%%%%%%%%%%%%%%%%%%%%%%%%%%%%%%%%%%%%%%%%%%%%%%%%%
%
%			Conclusion
%
%%%%%%%%%%%%%%%%%%%%%%%%%%%%%%%%%%%%%%%%%%%%%%%%%%%%

\section{Conclusion}
\label{sec:conclusion}
In this paper we introduce a class of local stochastic volatility models.  Within our modeling framework, we obtain a formula (written as an infinite series) for the price of any European option.  Additionally, we obtain an explicit expression for the implied volatility smile induced by our class of models.  As an example of our framework, we introduce a CEV-like model, which corrects one possible short-coming of the CEV model; namely, our choice of local volatility function does not drop to zero as the value of the underlying increases.  Finally, in the CEV-like example, we show that option prices can be computed with the same level of efficiency as other models in which option prices are computed as Fourier-type integrals.

\subsection*{Thanks}
The author would like to thank Bjorn Birnir and two anonymous reviewers for their helpful comments.

%%%%%%%%%%%%%%%%%%%%%%%%%%%%%%%%%%%%%%%%%%%%%%%%%%%%
%
%			Bibliography
%
%%%%%%%%%%%%%%%%%%%%%%%%%%%%%%%%%%%%%%%%%%%%%%%%%%%%

\clearpage
\bibliographystyle{chicago}
\bibliography{BibTeX-Master}	

%%%%%%%%%%%%%%%%%%%%%%%%%%%%%%%%%%%%%%%%%%%%%%%%%%%%%%%%%
%
%								APPENDIX
%
%%%%%%%%%%%%%%%%%%%%%%%%%%%%%%%%%%%%%%%%%%%%%%%%%%%%%%%%%

\clearpage
\appendix

%%%%%%%%%%%%%%%%%%%%%%%%%%%%%%%%%%%%%%%%%%%%%%%%%%%%%%%%%
%								Spectral Theory
%%%%%%%%%%%%%%%%%%%%%%%%%%%%%%%%%%%%%%%%%%%%%%%%%%%%%%%%%

\section{Spectral theory of normal operators in a Hilbert space}
\label{sec:spectral}
In this appendix we briefly summarize the theory of normal operators acting on a Hilbert space.  A detailed exposition on this topic (including proofs) can be found in \cite{reedsimon} and \cite{rudin1973functional}.
\par
Let $\Hc$ be a Hilbert space with inner product $(\cdot,\cdot)$.  
The \emph{adjoint} of an operator $\Ac$ acting in $\Hc$ is an operator $\Ac^{*}$ such that $(\Ac f,g)	=	(f, \Ac^{*} g)$.  Here, for simplicity, we have assumed $\text{dom}(\Ac)=\text{dom}(\Ac^{*})=\Hc$.
An operator $(\text{dom}(\Ac),\Ac)$ is said to be \emph{normal} in $\Hc$ if it is closed, densely defined and commutes with its adjoint: $\Ac^* \Ac = \Ac \Ac^*$.
\par
Suppose $\Ac$ is a normal operator acting on the Hilbert space $\Hc=L^2(\Rb,dy)$.  For any Borel measurable function $g$, the operator $g(\Ac)$ can be constructed as follows.  First, one solves the \emph{proper} and \emph{improper}
\footnote{The term ``improper'' is used because the improper eigenvalues $\lam \notin \sig_d(\Ac)$ and the improper eigenfunctions $\psi_\lam \notin \Hc$ since $\(\psi_\lam,\psi_\lam\)=\infty$.}
eigenvalue problems
\begin{align}
\text{proper:}&&
\Ac \, \psi_n
		&= 		\phi_n \, \psi_n , &
\phi_n 
		&\in 	\sig_d(\Ac) , &
\psi_n
		&\in	\Hc , \label{eq:proper} \\
\text{improper:}&&
\Ac \, \psi_\lam
		&= 		\phi_\lam \, \psi_\lam , &
\phi_\lam 
		&\in 	\sig_c(\Ac) , &
\psi_\lam
		&\notin	\Hc , \label{eq:improper}
\end{align}
where $\sig_d(\Ac)$ and $\sig_c(\Ac)$ denote the discrete and continuous spectrum of $\Ac$, respectively.  For the improper eigenvalue problem one extends the domain of $\Ac$ to include all functions $\psi$ for which $\Ac \psi$ makes sense and for which the following boundedness conditions are satisfied
\begin{align}
\lim_{y \to \pm \infty} |\psi(y)|^2 
		&< \infty . \label{eq:bound}
\end{align}
After normalizing, the proper and improper eigenfunctions of $\Ac$ satisfy the following orthogonality relations
\begin{align}
\( \psi_n, \psi_m \)
		&=	\del_{n,m} , & 
\( \psi_\lam, \psi_{\lam'} \)
		&=	\del(\lam - \lam') , &
\( \psi_n, \psi_\lam \) 
		&=	0 .
\end{align}
The operator $g(\Ac)$ is then defined as follows (see \cite{hanson2002operator}, section 5.3.2)
\begin{align}
g(\Ac) f
		&=	\sum_{\phi_n \in \sig_d(\Ac)}	g(\phi_n) \( \psi_n, f \) \psi_n +
				\int_{\sig_c(\Ac)} g(\phi_\lam) \( \psi_\lam, f \) \psi_\lam d\lam . \label{eq:operator}
\end{align}

%%%%%%%%%%%%%%%%%%%%%%%%%%%%%%%%%%%%%%%%%%%%%%%%%%
%
%		Proof 3
%
%%%%%%%%%%%%%%%%%%%%%%%%%%%%%%%%%%%%%%%%%%%%%%%%%

\section{Proof of Theorem \ref{thm:u.eps}}
\label{sec:proof3}
Our strategy is to show that $\Ac^\eps=\Ac_0 + \eps \, \eta \, \Ac_1$ generates a semigroup $\Pc_t^\eps = \exp(t \, \Ac^\eps)$.  This will guarantee that $u^\eps(t,y) = \Pc_t^\eps h(y)$ is an analytic function of $\eps$, which in turn, justifies the use of expansion \eqref{eq:u.eps.expand}.  Throughout this section we will work on the Hilbert space $\Hc=L^2(\Rb,dy)$.  
We let $\dom{\Ac_i}=\Sc$, the Schwartz space of rapidly decreasing functions on $\Rb$:
\begin{align}
\Sc
	&=		\{ f \in C^\infty(\Rb): \norm{f}_{\alpha,\beta}< \infty, \forall \, \alpha, \beta \} , &
\norm{f}_{\alpha,\beta}
	&=		\sup_{y \in \Rb} |y^\alpha \d^\beta f(y)| . 				\label{eq:schwartz}
\end{align}
We note that $\Sc$ is a dense subset of $\Hc$.  
{}{Thus, $\Ac_i$ has a unique extension $\overline{\Ac}_i$ with domain $\text{dom}(\overline{\Ac}_i)=\Hc$.}
Our analysis begins with a Theorem from \cite{chernoff}:
\begin{theorem}
Let $\Ac$ be the generator of a $C_0$ contraction semigroup $\Pc_t^0=\exp(t\, \Ac)$ on a Banach space.  Let $\eps \, \Bc$ be a dissipative operator with a densely defined adjoint.  Assume that the inequality
\begin{align}
\norm{\eps \, \Bc u}
	&\leq		c \norm{u} + b \norm{ \Ac u } , &
\forall \, u
	&\in		{\rm dom}(\Ac) ,
\end{align}
holds for some $c \geq 0$ and $b \leq 1$ (i.e., the operator $\eps \, \Bc$ is bounded {}{relative} to $\Ac$ with {}{relative} bound $b \leq 1$).  Then the closure of $\Ac^\eps:=\Ac+\eps\,\Bc$ generates a $C_0$ contraction semigroup $\Pc_t^\eps=\exp(t \, \Ac^\eps)$.
\end{theorem}
\begin{remark}\rm
\label{rmk:dissipative}
{}{Recall, an operator $\Ac$ is \emph{dissipative} if $\text{Re}(u,\Ac u) \leq 0$ for all $u \in \Hc$.}
\end{remark}
\begin{remark}\rm
The operator $\Ac_0$ is the generator of a $C_0$ contraction semigroup $\Pc_t^0 = \exp(t \, \Ac_0)$ on $\Hc$.  {}{Thus, we must (i) show that $\eps \, \eta \, \Ac_1$ has a densely defined adjoint, (ii) show that $\eps \, \eta \, \Ac_1$ is dissipative and (iii) derive conditions under which $\eps \, \eta \, \Ac_1$ is bounded {}{relative} to $\Ac_0$ with {}{relative} bound less than or equal to one.}
\end{remark}
\noindent
{}{To show (i) we note that the adjoint of $\eps \, \eta \, \Ac_1$, given by $\( \eps \, \eta \, \Ac_1 \)^*=\eps \, \Ac_1^* \, \eta$,  has domain $\dom{\eps \, \Ac_1^*\,  \eta}=\Sc$.  As mentioned above, $\Sc$ is densely defined in $\Hc$.  To show (ii), we note that, if an operator satisfies the positive maximum principle
\footnote{An operator $\Ac$ satisfies the \emph{positive maximum principle} if, for any function $f \in \dom{\Ac}$ that attains a maximum at $y^*$ such that $f(y^*) \geq 0$ we have $\Ac f(y^*) \leq 0$.}
then that operator is dissipative (see \cite{ethier1986markov}, Lemma 4.2.1 on page 165).  The following Theorem will be useful.}
\begin{theorem}
Let $\Ac$ be a linear operator with domain $\dom{\Ac}=\Sc$.  Then $\Ac$ satisfies the positive maximum principle if and only if
\begin{align}
\Ac
	&=	\frac{1}{2}a^2(y) \d^2 + b(y) \d 
			+ \int_\Rb \nu(y,dz) \( e^{z \d} - 1 - \Ib_{\{z < R\}} z \d \) 
			- c(y),
			\label{eq:operator.form}
\end{align}
for some $a(x) \geq 0$, $b(x) \in \Rb$, $c(x) \geq 0$, $R \in [0,\infty]$ and $\nu(y,dz)$ satisfying
\begin{align}
\int_\Rb \nu(y,dz) \( 1 \wedge z^2 \) < \infty .
\end{align}
Operators of the form \eqref{eq:operator.form} are called  \emph{L\'evy-type operators}.
\end{theorem}
\qquad \begin{proof}
See Theorem 2.12 of \cite{hoh1998pseudo}.
\qquad \end{proof}
\noindent
{}{The operator $\eps \, \eta \, \Ac_1$ is clearly of the form \eqref{eq:operator.form}.  Hence, $\eps \, \eta \, \Ac_1$ satisfies the positive maximum principle and is therefore dissipative.
Finally, for part (iii), the following Theorem gives conditions under which $\eps \, \eta \, \Ac_1$ is bounded relative to $\Ac_0$ with relative bound one.}
\begin{proposition}
Suppose $\eps \leq \frac{a^2 }{ \left\|\eta \right\|}$ (which is the condition given in Theorem \ref{thm:u.eps}).  Then $\eps \, \eta \, \Ac_1$ is bounded {}{relative} to $\Ac_0$ with {}{relative} bound less than or equal to one.
\end{proposition}
\begin{proof}
Clearly, for any $u \in {\rm dom}(\Ac_0)$ we have
\begin{align}
\left\|\eps \, \eta \, \Ac_1 \, u\right\| 
	&\leq \eps \left\|\eta \right\| \cdot \left\| \Ac_1 u \right\| 
	=				\frac{\eps}{a^2}  \left\|\eta \right\| \cdot \left\| \Ac_0 u\right\|
	\leq		 \left\| \Ac_0 u\right\| .
\end{align}
\qquad \end{proof}
\noindent
The proof of Theorem \ref{thm:u.eps} is complete.

%%%%%%%%%%%%%%%%%%%%%%%%%%%%%%%%%%%%%%%%%%%%%%%%%%%%
%
%			Figures
%
%%%%%%%%%%%%%%%%%%%%%%%%%%%%%%%%%%%%%%%%%%%%%%%%%%%%

%%%%%%%%%%%%%%%%%%%%%%%%%%%%%%%%%%%%%%%%%%%%%%%%%%%

\clearpage
\begin{figure}
\centering
\includegraphics[width=.95\textwidth,height=.5\textheight]{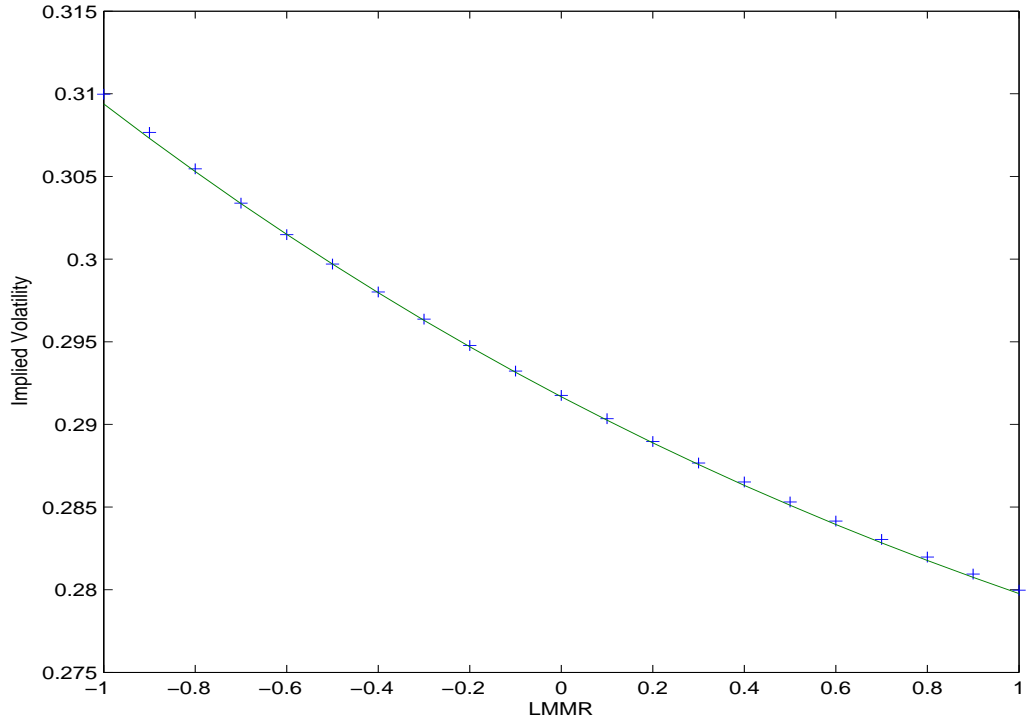}
\caption{{}{We compute $u^\eps$, the prices of set of European call options (i) by using approximation (\ref{eq:uN.approx}) with $N=10$ and (ii) by Monte Carlo simulation.  We then convert the obtained prices to implied volatilities by inverting Black-Scholes numerically.  The results of this procedure are plotted above.  The green line corresponds to implied volatilities computed using approximation (\ref{eq:uN.approx}).  The blue crosses corresponds to implied volatilities computed by Monte Carlo simulation.  The units of the horizontal axis are $\log$-moneyness-to-maturity ratio $\text{LMMR}:=(k-y)/t$.  The following parameters are used in these plots: $y=0.00$, $a=0.25$, $\sqrt{\eps}=0.15$, $\beta=-0.75$, $t=1.0$.  The two methods of computation produce indistinguishable implied volatilities.}}
\label{fig:MonteCarlo}
\end{figure}

%%%%%%%%%%%%%%%%%%%%%%%%%%%%%%%%%%%%%%%%%%%%%%%%%%%

\clearpage
\begin{figure}
\centering
\begin{tabular}{ | c | c |}
\hline
%{}&{} \\
$n=1$ & $n=2$ \\
%& \\
\includegraphics[width=.465\textwidth,height=.25\textheight]{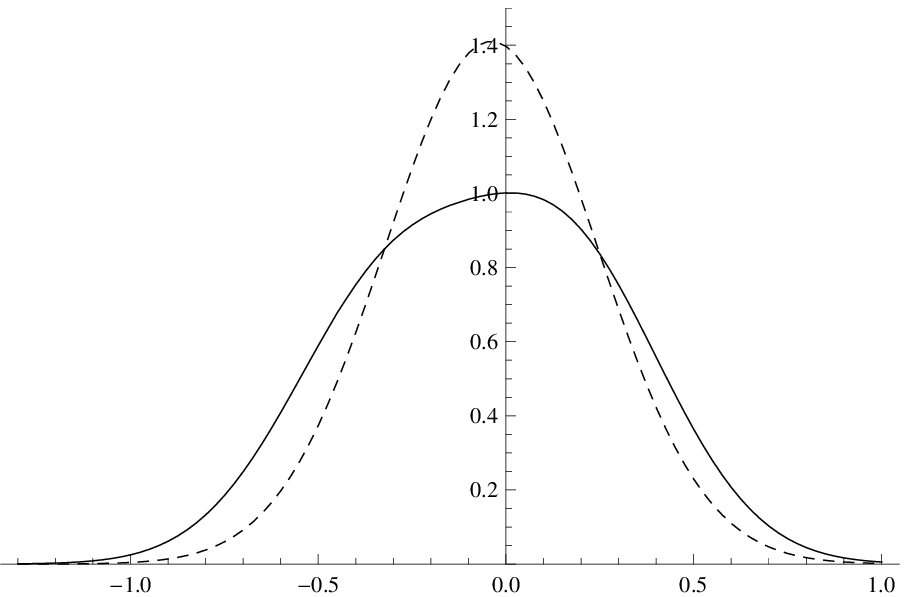} &
\includegraphics[width=.465\textwidth,height=.25\textheight]{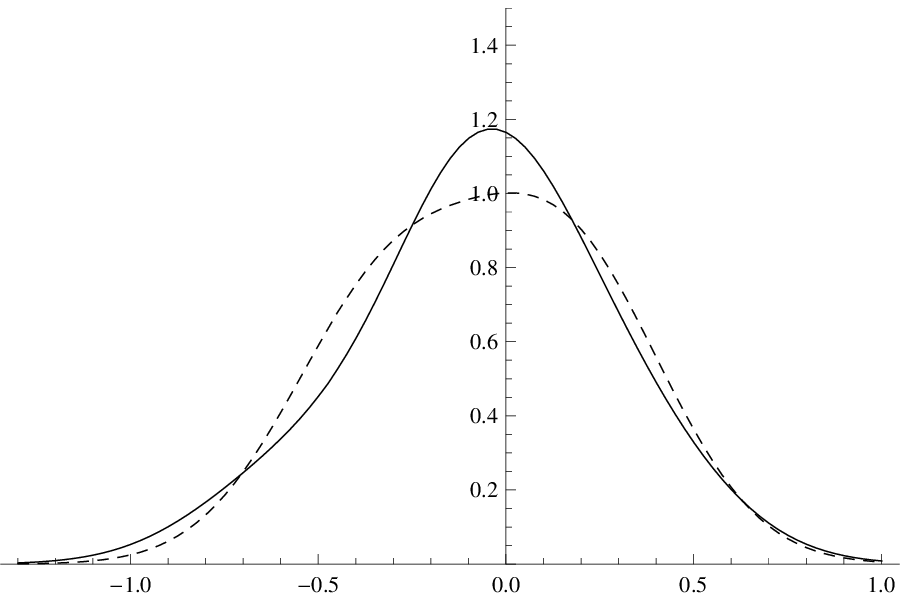} \\ \hline 
%& \\
$n=3$ & $n=4$ \\
%& \\
\includegraphics[width=.465\textwidth,height=.25\textheight]{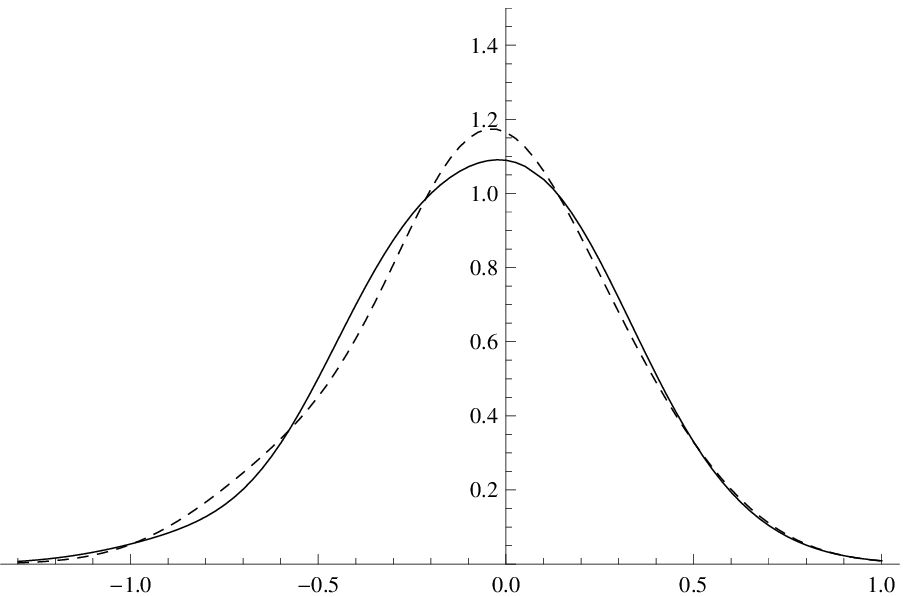} &
\includegraphics[width=.465\textwidth,height=.25\textheight]{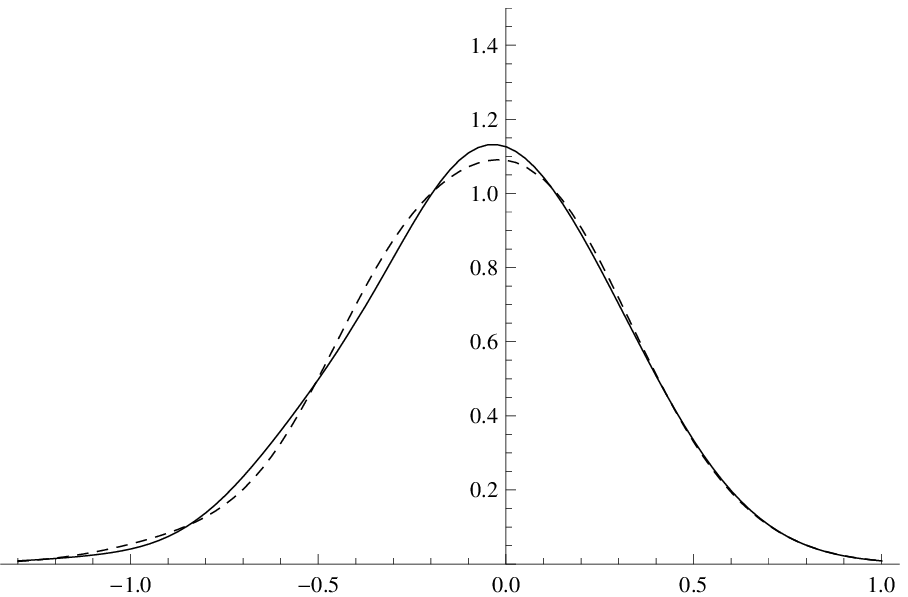} \\ \hline 
%& \\
$n=5$ & $n=6$ \\
%& \\
\includegraphics[width=.465\textwidth,height=.25\textheight]{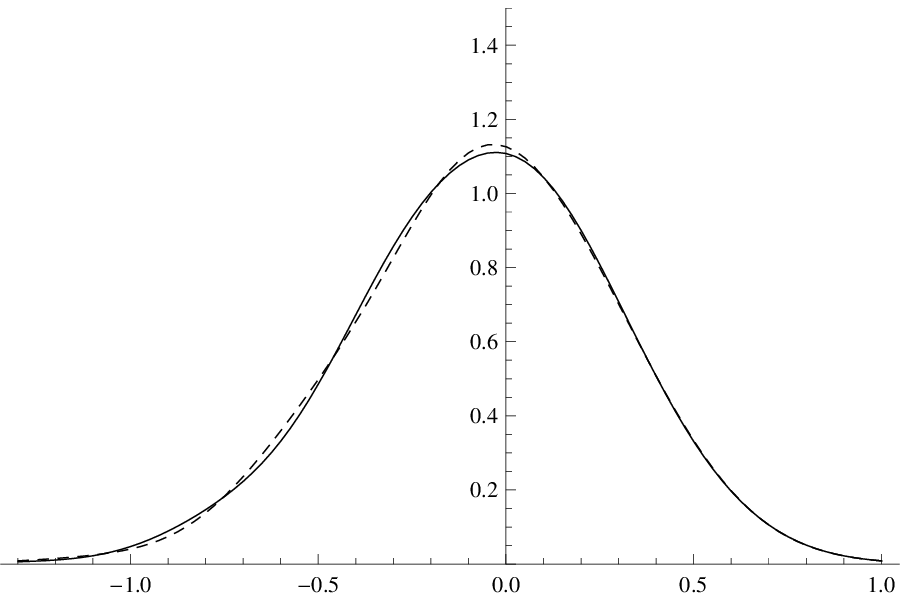} &
\includegraphics[width=.465\textwidth,height=.25\textheight]{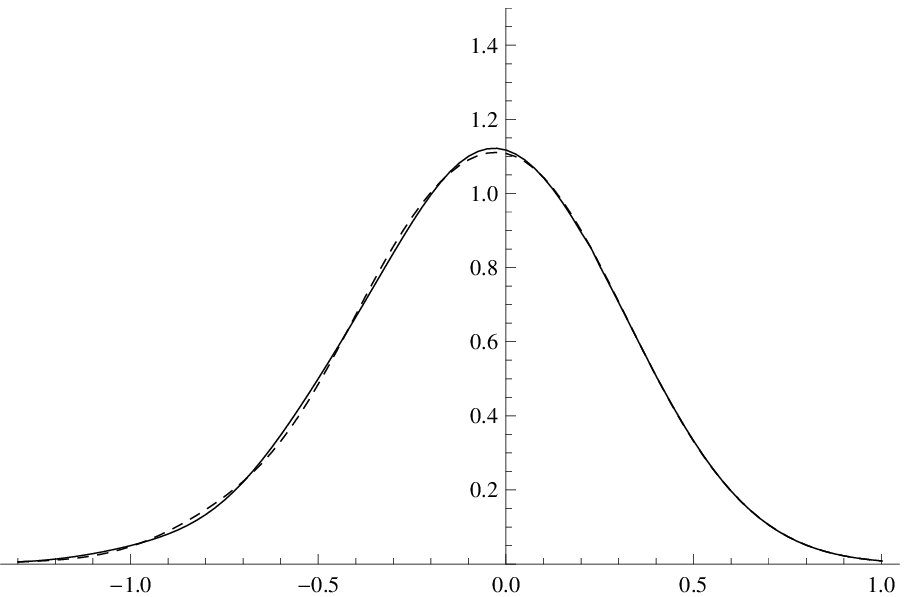} \\ \hline
\end{tabular}
\caption{{}{A plot of the approximate transition density $p^{(n)}(t,y;0)$ for different values of $n$.  In order to see convergence, we plot $p^{(n)}$ (solid) and $p^{(n-1)}$ (dashed) together.  We see almost no difference between $p^{(5)}$ and $p^{(6)}$ (lower right).}  Note that the density of $Y_t$ has a fat tail to the left, which is expected since the local volatility function $\sig(e^y) = (a^2 + \eps e^{\beta y} )^{1/2}$ increases as $y \to -\infty$.  The following parameters are used in these plots: $a=0.20$, $\sqrt{\eps}=0.15$, $\beta=-0.85$, $t=2.0$.}
\label{fig:density}
\end{figure}

\clearpage

%%%%%%%%%%%%%%%%%%%%%%%%%%%%%%%%%%%%%%%%%%%%%%%%%%%

\begin{figure}
\centering
\begin{tabular}{ | c | c |}
\hline
%& \\
$n=2$ & $n=3$ \\
%& \\
\includegraphics[width=.465\textwidth,height=.25\textheight]{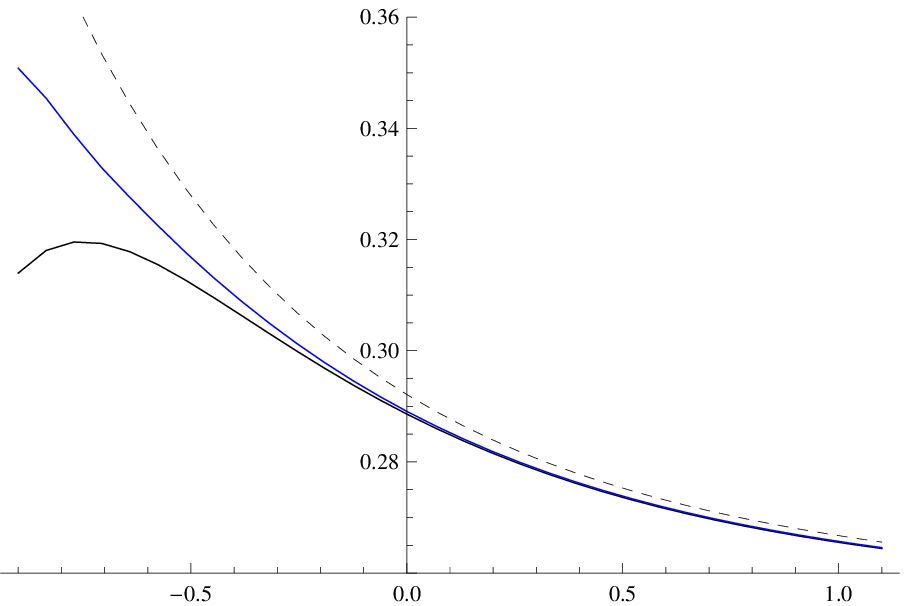} &
\includegraphics[width=.465\textwidth,height=.25\textheight]{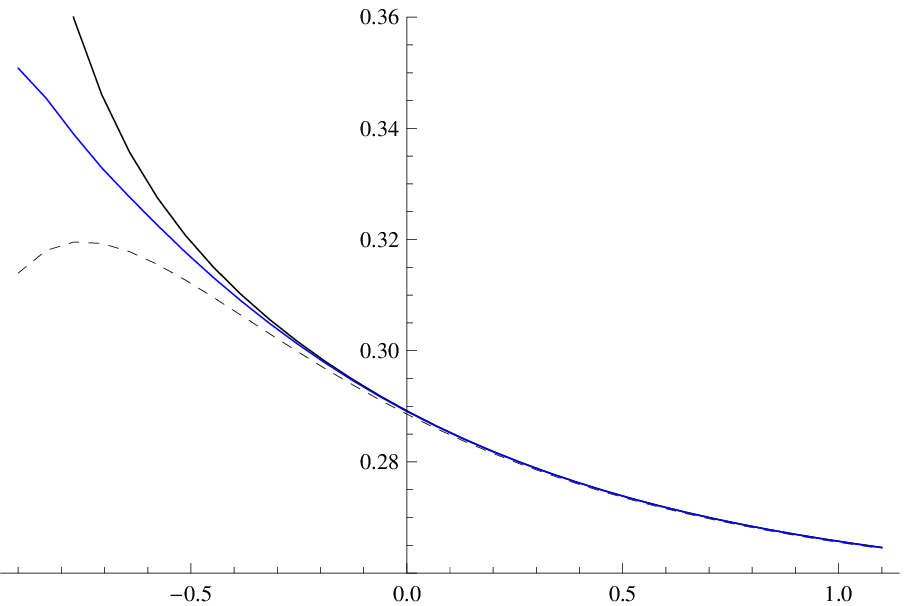} \\ \hline 
%& \\
$n=4$ & $n=5$ \\
%& \\
\includegraphics[width=.465\textwidth,height=.25\textheight]{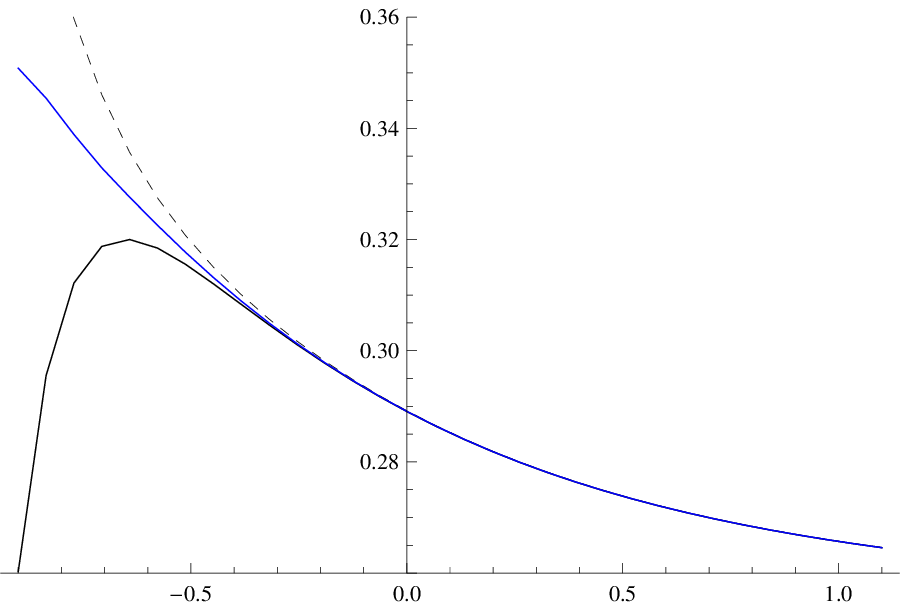} &
\includegraphics[width=.465\textwidth,height=.25\textheight]{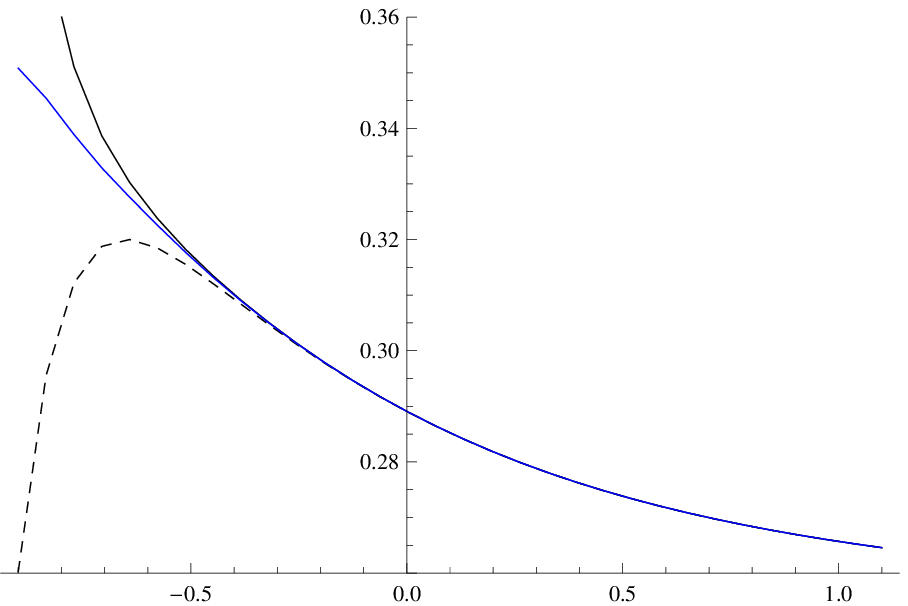} \\ \hline
\end{tabular}
\caption{{}{For different values of $n$, we plot $\sig^{(n)}$ (solid black), $\sig^{(n-1)}$ (dashed black) and $\sig^{\eps}$ (solid blue) as a function of LMMR.  For $\text{LMMR} > -0.5$ we see fast convergence of $\sig^{(n)}$ to $\sig^\eps$.  For $\text{LMMR} < -0.5$, however, convergence is quite slow.  Note that, although $\sig^{(n)}$ \emph{appears} to more closely approximate $\sig^\eps$ for odd $n$ than for even $n$, this is simply due to the fact that, for even $n$, $\sig^{(n)}$ diverges downward, whereas for odd $n$, $\sig^{(n)}$ diverges upward, matching the convexity of $\sig^\eps$.  In fact, the region of convergence, loosely defined as the set of LMMR for which $\sig^{(n)}$ closely approximates $\sig^\eps$, increases for every $n$.  
The following parameters are used in these plots: $a=0.25$, $\sqrt{\eps}=0.15$, $\beta=-0.75$, $t=3.0$ $y=0.1$.}}
\label{fig:impvol}
\end{figure}

\end{document}